\documentclass[11pt]{amsart}

\usepackage{fullpage}
\usepackage{latexsym,amssymb,amsfonts,amsmath,mathrsfs,float}
\usepackage[font=small,labelfont=bf, width=.618\textwidth]{caption} 
\usepackage{xcolor}
\usepackage{tikz, graphics, enumerate}
\usepackage{hyperref}

  \renewcommand{\Pr}{\mbox{\rm Pr}}

  \newcommand{\C}{\mathbb{C}} 
  \newcommand{\F}{\mathbb{F}} 

  \newcommand{\st}{:\,} 

  \newcommand{\eps}{\varepsilon}

  \newcommand{\beq}{\begin{equation}}
  \newcommand{\eeq}{\end{equation}}
  \newcommand{\beqn}{\begin{equation*}}
  \newcommand{\eeqn}{\end{equation*}}
  \newcommand{\beqr}{\begin{eqnarray}}
  \newcommand{\eeqr}{\end{eqnarray}}
  \newcommand{\beqrn}{\begin{eqnarray*}}
  \newcommand{\eeqrn}{\end{eqnarray*}}
  \newcommand{\bmline}{\begin{multline}}
  \newcommand{\emline}{\end{multline}}
  \newcommand{\bmlinen}{\begin{multline*}}
  \newcommand{\emlinen}{\end{multline*}}

  \theoremstyle{plain}
  \newtheorem{theorem}{Theorem}

  \newtheorem{claim}[theorem]{Claim}

  \theoremstyle{definition}
  \newtheorem{definition}[theorem]{Definition}

  \theoremstyle{remark}
  \newtheorem{remark}[theorem]{Remark}
  \renewenvironment{proof}[1][]{
    	\begin{trivlist}
     	\item[\hspace{\labelsep}{\em\noindent Proof#1:\/}]}
     	{{\hfill$\Box$}
    	\end{trivlist}
  }
  
  \makeatletter
  \newtheorem*{rep@theorem}{\rep@title}
  \newcommand{\newreptheorem}[2]{%
  \newenvironment{rep#1}[1]{%
  \def\rep@title{#2 \ref{##1}}%
  \begin{rep@theorem}}%
  {\end{rep@theorem}}}
  \makeatother

  \newreptheorem{lemma}{Lemma}
  \newreptheorem{theorem}{Theorem}
  \newreptheorem{corollary}{Corollary}
  \newreptheorem{proposition}{Proposition}
  \newreptheorem{conjecture}{Conjecture}
  \newreptheorem{example}{Example}

\begin{document}\title{Improved Lower Bounds for the Restricted Isometry Property of Subsampled Fourier Matrices}
\author{
       Shravas Rao}
\thanks{Supported by the Simons Collaboration on Algorithms and
Geometry}
\address{Courant Institute, New York University, 251 Mercer Street, New York NY 10012, USA}
\email{rao@cims.nyu.edu}
\date{\today}

\maketitle
\begin{abstract}
Let $A$ be an $N \times N$ Fourier matrix over $\F_p^{\log{N}/\log{p}}$ for some prime $p$.
We improve upon known lower bounds for the number of rows of $A$ that must be sampled so that the resulting matrix $M$ satisfies the restricted isometry property for $k$-sparse vectors.
This property states that $\|Mv\|_2^2$ is approximately $\|v\|_2^2$ for all $k$-sparse vectors $v$.
In particular, if $k = \Omega( \log^2{N})$, we show that $\Omega(k\log{k}\log{N}/\log{p})$ rows must be sampled to satisfy the restricted isometry property with constant probability.
\end{abstract}

\section{Introduction}

We say that the matrix $M \in \C^{q \times N}$ satisfies the $(k, \eps)$-restricted isometry property if for all vectors $v \in \C^{N}$ with at most $k$ non-zero entries,
\[
(1-\eps)\|v\|_2^2 \leq \|M v\|_2^2 \leq (1+\eps)\|v\|_2^2.
\]
This notion is due to Cand\`es and Tao~\cite{CT05} and has found many applications, especially in the field of compressed sensing~\cite{C08}.
One goal is to construct such matrices $M$ with as few rows as possible.

It is well known that if the entries of $M$ are chosen independently at random, then $q = \Theta(k \log{N/k})$ rows are both necessary~\cite{FPRU10} and sufficient~\cite{CT06, BDDW08, MPTJ08} to obtain a matrix that satisfies the restricted isometry property.
Recently, much attention has been placed on matrices obtained by sampling from the rows of the Fourier matrices.
In particular if the vector-space is $\F_p^{\log{N}/\log{p}}$ for some prime $p$, this is the matrix defined by $A_{i, j} = \omega^{\langle i, j\rangle}$ where $\omega$ is the $p$th root of unity.
In the case that $p = 2$, the corresponding matrix is sometimes referred to in the literature as a Hadamard matrix.
Such matrices have the advantage that matrix vector multiplication can be performed quickly.
However, because the entries are now highly dependent on each other, it is not clear if one can still sample a small number of rows and obtain matrices that satisfy the restricted isometry property.

Cand\`es and Tao gave the first bound of $O(k \log^6{N})$ rows being sufficient to satisfy the restricted isometry property~\cite{CT06}.
This bound was improved by Rudelson and Vershynin~\cite{RV08} to $O(k \log^2{k} \log{(k\log{N})} \log{N})$, by  Cheraghchi, Guruswami, and Velingker~\cite{CGV13} to $O(k \log^3{k}\log{N})$, and by Haviv and Regev~\cite{HR17} to $O(k \log^2{k} \log{N})$.
Additionally, a bound of $O(k \log{k} \log^2{N})$ was proved by Bourgain~\cite{B14}.

On the other side, a lower bound of $\Omega(k \log N)$ rows was shown by Bandeira, Lewis, and Mixon~\cite{BLM18}.
In this work we obtain the following lower bound.

\begin{theorem}\label{lem:maininformal}
Let $A$ be the $N \times N$ matrix defined by $A_{i, j} = \omega^{\langle i, j\rangle}$ for $i, j \in \F_p^{\log_p{N}}$ where $\omega$ is the $p$th root of unity.
For some $q = O(k\log{k}\log{N}/\log{p})$, let $M \in \C^{q \times N}$ be a matrix whose $q$ rows are sampled uniformly and independently at random from the rows of $A$.
If $k = \Omega(\log^2{N})$ and $k \leq N^{1/3}$, then $M$ does not satisfy the $(k, \eps)$-restricted isometry property with constant probability for any $\eps$.
\end{theorem}

A more technical version of Theorem~\ref{lem:maininformal} can be found in Section~\ref{sec:main}, that also discusses the case of general $k$.

The techniques used to prove Theorem~\ref{lem:maininformal} are similar to the techniques used in~\cite{BLM18}.
The authors of~\cite{BLM18} noticed an instance of the coupon collecting problem embedded in the problem of whether or not a subsampled Fourier matrix satisfies the restricted isometry property.
In this paper, we find many instances of the coupon collecting problem and combine them to obtain a stronger lower bound.

\subsection{Related Work}

Jaros{\l}aw B{\l}asiok, Patrick Lopatto, Kyle Luh, and Jake Marcinek have simultaneously and independently proved a similar result and we refer the reader to their forthcoming preprint for the details.  

\section{Notation}

Unless otherwise noted, we let $\log{x}$ denote the natural logarithm.

Let $A$ be the $N \times N$ matrix whose entries are defined by $A_{i, j} = \omega^{\langle i, j\rangle}$ where $i$ and $j$ are vectors in $\F_p^{\log_p{N}}$ and $\omega$ is the $p$th root of unity.

If $T$ and $S$ are sequences with elements from $[N]$, then we let $A_{T, \cdot}$ be the matrix whose rows are from $T$, $A_{\cdot, S}$ be the matrix whose columns are from $S$, and $A_{T, S} = (A_{T, \cdot})_{\cdot, S} = (A_{\cdot, S})_{\cdot, T}$.


Finally, we define the concept of \emph{shattering} that will be used in the proof of Theorem~\ref{lem:maininformal}.

\begin{definition}
Let $V \subseteq \F_p^{\log_p{N}}$ be a subspace and let $P_V: \F_p^{\log_p{N}} \rightarrow V$ be the projection of $\F_p^{\log_p{N}}$ onto $V$.
We say that a sequence $Q$ whose elements are from the set $\F_P^{\log_p{N}}$ \emph{shatters} a subspace $V$ if
\[
P_V(Q) := \{P_V(r) \st r \in Q\}
\]
is equal to $V$.
\end{definition}

\section{Proof of Main Result}\label{sec:main}

%
%
%
%
%
%

We prove the following bound, which is the main result of this paper.
This statement is a more technical version of Theorem~\ref{lem:maininformal}.

\begin{theorem}\label{lem:main}
Let $A$ be the $N \times N$ matrix defined by $A_{i, j} = \omega^{\langle i, j\rangle}$ for $i, j \in \F_p^{\log_p{N}}$ where $\omega$ is the $p$th root of unity, and let $r_1, \ldots, r_{|Q|}$ be independent and identically distributed random variables uniform over $[N]$ so that $Q = (r_1, \ldots, r_{|Q|})$ is a sequence.
Then there exists a universal constant $C < 1$ such that if $k \leq N/2$
\[
|Q| \leq (k+\sqrt{k})((\log_p{k}/2+1)(\log{N}-3\log{2k})+\log{k}),
\] the matrix $A_{Q, \cdot}$ does not satisfy the restricted isometry property with probability at least 
\[
C\exp\left(-|Q|\frac{2\sqrt{k}}{k^2-k}-1\right) \geq \frac{C}{e}\exp\left(\frac{-2-2\log{k}}{\sqrt{k}-1}\right)\left(\frac{N}{2k^3}\right)^{-2(\log_p{k/2}+1)/(\sqrt{k}-1)}
\]
\end{theorem}

We start by proving the following necessary property for $A_{Q, \cdot}$ to satisfy the restricted isometry property.
\begin{claim}\label{claim:rowsequal}
Let $V \subseteq \F_p^{\log_p{N}}$ be a subspace with dimension $\log_p{k}$.
Then if $A_{Q, \cdot}$ satisfies the restricted isometry property, then $Q$ shatters $V$.
\end{claim}
\begin{proof}
Let $v_1, \ldots, v_{\log_p{k}}$ be a basis for $V$.
Note that all columns of $A_{\cdot, V}$ can be determined from the columns $A_{\cdot, \{v_1\}}, \ldots, A_{\cdot, \{v_{\log_p{k}}\}}$ and thus, $A_{\cdot, V}$ contains $k$ distinct rows.
Additionally, $A_{\cdot, V}$ has rank $k$ as the matrix $A$ itself is an isometry.
To see this, for any vector $v$, let $v_{V}$ be the restriction of $v$ to the coordinates in $V$.
If $v$ is a $k$-sparse vector that is non-zero only in the coordinates corresponding to $V$, and $v$ is non-zero, then $A v$ is equal to $A_{\cdot, V} v_{V}$ and both are non-zero.
Thus, if the set of rows of $A_{Q, V}$ is not equal to the set of rows of $A_{\cdot, V}$, it does not have full rank, and there exists a $k$-sparse vector $v$ for which $A_{Q, \cdot}(v) = 0$ and $A_{Q, \cdot}$ does not satisfy the restricted isometry property.


Thus, for every $w \in \F_p^{\log_p{k}}$, there must exist an $r \in Q$ such that $A_{\{r\}, V} = A_{\{w\}, V}$.
Let $P_V: \F_p^{\log_p{N}} \rightarrow V$ be the projection of $\F_p^{\log_p{N}}$ onto $V$.
Then $P_V(r) = P_V(w)$, and thus $\{P_V(r) \st r \in Q\}$ is equal to $V$ as desired.

\end{proof}

To prove Theorem~\ref{lem:main} we fix many subspaces of $\F_p^{\log_p{N}}$ of dimension $\log_p{k}$ and analyze the probability that $Q$ does not shatter these subspaces.
In particular, we will fix these subspaces so that the event that $Q$ does not shatter any given subspace is close to being independent of the event that $Q$ does not shatter any other given subspace.
This allows us to fix enough subspaces such that the probability that $Q$ does not shatter any subspaces is high.

The probability that $Q$ does not shatter $V$ for a subspace $V$ of dimension $\log_p{k}$ can be analyzed as in the coupon collector's problem.
This probability is approximately $k(1-1/k)^{|Q|} \approx \exp(-(|Q|-\log{k})/k)$, but we will use the following lower bound.

\begin{claim}\label{claim:indprob}
Let $V \subseteq \F_p^{\log_p{N}}$ be a subspace of dimension $\log_p{k}$ and let $P_V: \F_p^{\log_p{N}} \rightarrow V$ be the projection of $\F_p^{\log_p{N}}$ onto $V$.
Let $Q = (r_1, \ldots, r_{|Q|})$ be a random sequence where $r_1, \ldots, r_{|Q|}$ are independent and identically distributed random variables uniform over $\F_p^{\log_p{N}}$.
Then
\[
\Pr\left[P_V(Q) \neq V\right] \geq
 k\left(1-\frac{1}{k}\right)^{|Q|}\left(1-k\left(1-\frac{1}{k-1}\right)^{|Q|}\right)
\]
as desired.
\end{claim}
\begin{proof}
Note that for $a, b \in V$ such that $a$ and $b$ are distinct,
\begin{align*}
\Pr[a \not\in P_V(Q) \text{ and } b \not\in P_V(Q)] 
&=
\Pr[a \not\in P_V(Q)]\Pr[b \not\in P_V(Q) \mid a \not\in P_V(Q)] \\
&= 
\left(1-\frac{1}{k}\right)^{|Q|}\left(1-\frac{1}{k-1}\right)^{|Q|}
\end{align*}
where the last equality follows by noting that if $a$ is not contained in $P_V(Q)$, the projection of each $r \in Q$ onto $V$ is a uniform random variable over $V \backslash \{a\}$ which has size $k-1$.
By the principle of inclusion-exclusion,
\begin{align*}
\Pr\left[P_V(Q) \neq V\right] 
&\geq 
\sum_{a \in  V}\Pr[a \not\in P_V(Q)] - \sum_{a, b \in  V}\Pr[a \not\in P_V(Q) \text{ and } b \not\in P_V(Q)] \\
&\geq k\left(1-\frac{1}{k}\right)^{|Q|}-k^2\left(1-\frac{1}{k}\right)^{|Q|}\left(1-\frac{1}{k-1}\right)^{|Q|}
\end{align*}
\end{proof}

The following claim gives a bound on the probability that for two subspaces $V_1$ and $V_2$, $Q$ does not shatter either, assuming that their intersection has small dimension.
This bound will be close to the bound in the case that the intersection $V_1$ and $V_2$ is $\{0\}$, i.e. the event that $Q$ shatters $V_1$ is independent of the event that $Q$ shatters $V_2$.

\begin{claim}\label{claim:intprob}
Let $V_1, V_2 \subseteq \F_p^{\log_p{N}}$ be two subspaces of dimension $\log_p{k}$ so that $\dim(V_1 \cap V_2) = m$ and let $P_{V_1}: \F_p^{\log_p{N}} \rightarrow V_1$ and $P_{V_2}: \F_p^{\log_p{N}} \rightarrow V_2$  be the projections of $\F_p^{\log_p{N}}$ onto $V_1$ and $V_2$ respectively.
Let $Q = (r_1, \ldots, r_{|Q|})$ be a random sequence where $r_1, \ldots, r_{|Q|}$ are independent and identically distributed random variables uniform over $\F_p^{\log_p{N}}$.
Then
\[
\Pr\left[P_{V_1}(Q) \neq V_1 \text{ and  }P_{V_2}(Q) \neq V_2\right] \leq
 k^2 \left(1-\frac{1}{k}\right)^{|Q|} \exp\left(-|Q|\frac{k-p^m}{(k-1)k}\right)
\]
\end{claim}
\begin{proof}
The left-hand side is bounded above by
\begin{align}
\sum_{a \in V_1, b \in V_2} &\Pr[a \not\in P_{V_1}(Q) \text{ and } b\not\in P_{V_2}(Q)] \nonumber \\
&=
\sum_{a \in V_1, b \in V_2} \Pr[a \not\in P_{V_1}(Q)]\Pr[ b\not\in P_{V_2}(Q) \mid a \not\in P_{V_1}(Q)] \nonumber \\
&= \sum_{a \in V_1, b \in V_2} \left(1-\frac{1}{k}\right)^{|Q|}\Pr[ b\not\in P_{V_2}(Q) \mid a \not\in P_{V_1}(Q)] \label{eq:intbound}
\end{align}
where the last equality follows from the fact that $V_1$ has $k$ vectors, and the size of $P_{V_1}^{-1}(v)$ is the same for all $v \in V_1$.

Given that $a$ is not contained within $P_{V_1}(Q)$, we can assume that for each $r_i \in Q$, $P_{V_1}(r_i)$ is uniformly distributed amongst all of the $k-1$ elements in $V_1 \backslash \{a\}$.

Let $P_{V_1 \cap V_2}: \F_p^{\log_p{N}} \rightarrow V_1\cap V_2$ be the projection of $\F_p^{\log_p{N}}$ onto $V_1 \cap V_2$.
If $P_{V_1 \cap V_2}(a) = P_{V_1 \cap V_2}(b)$, then of the now $k-1$ choices for $P_{V_1}(r_i)$ for each $r_i$, there are $k/p^m-1$ that can lead to $b = P_{V_2}(r_i)$.
Conditioned on any of the $k/p^m-1$ choices, the probability that in fact $P_{V_2}(r_i)$ is $b$ is equal to $p^m/k$.
Thus,
\[
\Pr[ b\not\in P_{V_2}(Q) \mid a \not\in P_{V_1}(Q)] = \left(1-\frac{k-p^m}{(k-1)p^m}\cdot \frac{p^m}{k} \right) =  \left(1-\frac{k-p^m}{(k-1)k}\right) 
\]
If $P_{V_1 \cap V_2}(a) \neq P_{V_1 \cap V_2}(b)$, then of the now  $k-1$ choices for $P_{V_1}(r_i)$ for each $r_i$, there are $k/p^m$ that can lead to $b = P_{V_2}(r_i)$.
Again, conditioned on any of the $k/p^m$ choices, the probability that $P_{V_2}(r_i)$ is $b$ is equal to $p^m/k$.
Thus,
\[
\Pr[ b\not\in P_{V_2}(Q) \mid a \not\in P_{V_1}(Q)] = \left(1-\frac{k}{(k-1)p^m}\cdot \frac{p^m}{k} \right) =  \left(1-\frac{1}{k-1}\right) 
\]

For every $a$, there are $k/2^m$ possible $b$ such that $a_v = b_v$ for all $v \in V_1 \cap V_2$.
Thus for each $a$,
\begin{align*}
 \sum_{ b \in V_2} \Pr[ b\not\in P_{V_2}(Q) \mid a \not\in P_{V_1}(Q)] 
&= 
\frac{k}{p^m} \left(1-\frac{k-p^m}{(k-1)k}\right)^{|Q|} + \frac{kp^m- k}{p^m}  \left(1-\frac{1}{k-1}\right)^{|Q|} \\
&\leq
\frac{k}{p^m} \exp\left(-|Q|\frac{k-p^m}{(k-1)k}\right) + \frac{kp^m- k}{p^m}  \exp\left(-|Q|\frac{1}{k-1}\right) \\
&=
k\left(\frac{1}{p^m} \exp\left(-|Q|\frac{k-p^m}{(k-1)k}\right) + \frac{p^m- 1}{2^m}  \exp\left(-|Q|\frac{1}{k-1}\right)\right) \\
&\leq
k \exp\left(-|Q|\frac{k-p^m}{(k-1)k}\right).
\end{align*}

Finally, we can plug the above into Eq.~\eqref{eq:intbound} to obtain

\[
\sum_{a \in V_1} k \left(1-\frac{1}{k}\right)^{|Q|}\exp\left(-|Q|\frac{k-p^m}{(k-1)k}\right) = k^2\left(1-\frac{1}{k}\right)^{|Q|}\exp\left(-|Q|\frac{k-p^m}{(k-1)k}\right)
\]
as desired.
\end{proof}

The following claim shows that there exist many subspaces of $\F_p^{\log_p{N}}$ so that every pair of subspaces has small intersection.

\begin{claim}\label{claim:manysub}
If $k \leq N/2$, then there exists a collection of subspaces $\mathcal{V} = \{V_1, \ldots, V_{\ell}\}$ of $\F_p^{\log_p{N}}$ where
\[
\ell =  \left(\frac{N}{2k^3}\right)^{\log_p{k}/2+1},\] so that $\dim(V_i) = \log_p{k}$ for all $i$, and $\dim(V_i \cap V_j) \leq \log_p{k}/2$.
\end{claim}
\begin{proof}
The number of subspaces of $\F_p^{\log_p{N}}$ of size $k$ is equal to
\[
A = \frac{(N-1)(N-p)(N-p^2)\cdots(N-k/p)}{(k-1)(k-p)(k-p^2)\cdots (k-k/p)}.
\]
This can been seen by first counting the number of sets of $k$ linearly independent vectors, and dividing by the number of bases of a given subspace of dimension $k$.

Fix a subspace $V$.
The number of subspaces $V'$ of $\F_p^{\log_p{N}}$ so that $\dim(V \cap V') > \log_p{k}/2$ is equal to
\begin{align*}
B &= \sum_{i=\log_p{k}/2+1}^{\log_p{k}} \frac{(k-1)(k-p)(k-p^2)\cdots (k-p^{i-1})}{(p^i-1)(p^i-2)(p^i-4)\cdots (p^i-p^{i-1})} \cdot
\frac{(N-k)(N-pk)(N-p^2k)\cdots(N-k^2/p^{i+1})}{(k/p^i-1)(k/p^i-2)(k/p^i-4)\cdots (k/p^i-k/p^{i+1})} \\
&\leq \log{k} \cdot \frac{(k-1)(k-p)\cdots (k-k^{1/2})}{(pk^{1/2}-1)(pk^{1/2}-2)\cdots (pk^{1/2}-k^{1/2})} \cdot
\frac{(N-k)(N-pk)\cdots(N-k^{3/2}/p)}{(k^{1/2}/p-1)(k^{1/2}/p-2) \cdots (k^{1/2}/p-k^{1/2}/p^2)}
\end{align*}
This follows by computing for each $i$, the number of subspaces in $V$ with dimension $i$, and the number of subspaces in $\F_p^{\log_p{N}} \backslash V$ with dimension $\log_p{k}-i$.
Thus, we can let $\ell$ be any integer less than $A/B$, and as long as $k \leq N/2$
\[
\frac{A}{B} \geq \frac{(N-k)^{\log_p{k}/2+1}}{k^{3(\log_p{k}/2+1)}} \geq \left(\frac{N}{2k^3}\right)^{\log_p{k}/2+1}
\]
as desired.
\end{proof}

Finally, we prove Theorem~\ref{lem:main}.

\begin{proof}[ of Theorem~\ref{lem:main}]
Let $\{V_1, \ldots, V_{\ell}\}$ be a subset of subspaces from Claim~\ref{claim:manysub}, where $\ell = \exp(|Q|/(k+\sqrt{k}))/(e k)$.
This is possible as long as $|Q| \leq (k+\sqrt{k})((\log_p{k}/2+1)(\log{N}-3\log{2k})+1+\log{k})$ because the total number of subspaces from Claim~\ref{claim:manysub} is $(N/(2k^3))^{\log_p{k}/2+1}$.
Additionally, for each $i \in [\ell]$, let $P_{V_i}: \F_p^{\log_p{N}} \rightarrow V_i$ be the projection of $\F_p^{\log_p{N}}$ onto $V_i$.

If for any $i \in [\ell]$, we have that $Q$ does not shatter $V_i$, then by Claim~\ref{claim:rowsequal}, $A_{Q}$ does not satisfy the restricted isometry property.
By the principle of inclusion-exclusion, and Claims~\ref{claim:indprob} and~\ref{claim:intprob},
\begin{align*}
\Pr&\left[\bigvee_{i \in [\ell]}P_{V_i}(Q) \neq V_i \right] \\
&\geq 
\sum_{i \in [\ell]} \Pr[P_{V_i}(Q) \neq V_i] - \sum_{i, j \in [\ell]} \Pr[P_{V_i}(Q) \neq V_i \text{ and } P_{V_j}(Q) \neq V_j] \\
&\geq 
\ell  k\left(1-\frac{1}{k}\right)^{|Q|}\left(1-k\left(1-\frac{1}{k-1}\right)^{|Q|}\right) - \ell^2  k^2 \left(1-\frac{1}{k}\right)^{|Q|} \exp\left(-|Q|\frac{k-\sqrt{k}}{(k-1)k}\right) \\
&=  \ell k \left(1-\frac{1}{k}\right)^{|Q|}\left(1-k\left(1-\frac{1}{k-1}\right)^{|Q|}-\ell  k \exp\left(-|Q|\frac{1}{k+\sqrt{k}}\right)\right). \\
&= C \exp\left(|Q|\frac{1}{k+\sqrt{k}}-1\right) \left(1-\frac{1}{k}\right)^{|Q|} \\
&\geq C \exp\left(|Q|\frac{1}{k+\sqrt{k}}-|Q|\frac{1}{k-\sqrt{k}}-1\right) \\
&= C \exp\left(-|Q|\frac{2\sqrt{k}}{k^2-k}-1\right)
\end{align*}
where the second equality follows by our choice of $\ell$, and the subsequent inequality follows from the fact that $1-1/k \geq \exp(-1/(k-\sqrt{k}))$ for $k \geq 1$.
\end{proof}

\begin{remark}
Note that if $k = \Omega(\log(N)^2)$, then the lower bound in Theorem~\ref{lem:main} is a constant, but one bounded above by $C/e$.
Let this constant be $C_1$.
We can obtain lower bounds approaching $1$ with just a constant factor decrease in $|Q|$ using the following argument, keeping in mind that the lower bound in Theorem~\ref{lem:main} is also a lower bound on the probability that $Q$ does not shatter every subspace of dimension $\log_p{k}$.

Consider $s$ independent random sequences $Q_1, \ldots, Q_s$ of equal size so that $|Q_1|+\cdots+|Q_s| \leq  (k+\sqrt{k})((\log_p{k}/2+1)(\log{N}-\log{2k})+\log{k})$.
The probability that ${Q_1, Q_2, \ldots, Q_s}$ does not shatter every subspace of dimension $\log_p{k}$ is at least $C_1$.
This event implies that ${Q_i}$ does not shatter every subspace of dimension $\log_p{k}$ for all $i$.
Because these events are independent, the probability that $Q_i$ shatters every subspace of dimension $\log_p{k}$, and thus does not satisfy the restricted isometry property for any fixed $i$ is $C_1^{1/s}$.
\end{remark}

\subsection*{Acknowledgments} I would like to thank Oded Regev and Ishay Haviv for their valuable comments.

\bibliographystyle{alphaabbrv}
\bibliography{riplb}
\end{document}